\documentclass[11pt]{article}
\usepackage[letterpaper, margin=1in]{geometry}

\usepackage{amsfonts,amsmath,amssymb,amsthm}
\usepackage{cite,comment,enumerate,hyperref}
\usepackage[ruled,vlined,linesnumbered]{algorithm2e}
\usepackage{thmtools}
\usepackage{thm-restate}

\usepackage[utf8]{inputenc} 
\usepackage[T1]{fontenc}    
\usepackage{bbm}

\newtheorem{definition}{Definition}
\newtheorem{lemma}{Lemma}
\newtheorem{corollary}{Corollary}

\newtheorem{theorem}{Theorem}

\newcommand{\vol}{\ensuremath{\text{vol}}}
\newcommand{\ncut}{\ensuremath{\text{NCut}}}
\newcommand{\nassoc}{\ensuremath{\text{NAssoc}}}
\newcommand{\nmod}{\ensuremath{\text{NMod}}}
\newcommand{\mmod}{\ensuremath{\text{Mod}}}

\renewcommand{\deg}{\ensuremath{\text{deg}}}

\newcommand{\Opt}{\ensuremath{\text{Opt}}}

\newcommand{\cquality}{q}
\newcommand{\E}{\mathbb{E}}

\title{Constant Approximation for Normalized Modularity\\ and Associations Clustering}

\author{%
 {Jakub Łącki}\\ \url{jlacki@google.com}\\
  Google Research
 \and
 {Vahab Mirrokni}\\ \url{mirrokni@google.com}\\
  Google Research%
 \and
 {Christian Sohler}\\ \url{sohler@cs.uni-koeln.de}\\
 University of Cologne
}

\begin{document}

\maketitle

\begin{abstract}%
We study the problem of graph clustering under a broad class of objectives in which the quality of a cluster is defined based on the ratio between the number of edges in the cluster, and the total weight of vertices in the cluster.
We show that our definition is closely related to popular clustering measures, namely normalized associations, which is a dual of the normalized cut objective, and normalized modularity.
We give a linear time constant-approximate algorithm for our objective, which implies the first constant-factor approximation algorithms for normalized modularity and normalized associations.
\end{abstract}


\section{Introduction}
Clustering is a fundamental task in unsupervised learning.
Many of the clustering problems are presented as optimization problems, in which the goal is to optimize a certain objective function over the space of all clusterings.\footnote{A notable exception is the hierarchical agglomerative clustering problem.}
From this point of view, one of the most fundamental theoretical questions is how efficiently and accurately we can solve the problem in question.
In particular, what is the best approximation ratio that we can achieve and what is the running time of the corresponding algorithm.

This question has been addressed, with varying success, for a number of clustering objectives.
For example, in the case of the widely-popular modularity objective~\cite{newman} any polynomial-time multiplicative approximation has been ruled under the $\textsc{P $\neq$ NP}$ conjecture~\cite{apx-mod}.
At the same time, while additive approximation can be computed in polynomial time, the corresponding algorithm is computationally expensive, as it is based on semidefinite programming~\cite{apx-mod}.

The hardness of developing efficient approximation algorithms appears to be a common feature among many clustering formulations.
This is likely caused by the fact that clustering objective functions are not designed to be easy to optimize, but rather to capture the desired properties of the final solution.
As a result, for a number of graph clustering problems the best known algorithms are either computationally expensive or do not provide good approximation guarantees.

For example, the minimizing disagreements formulation of correlation clustering admits a constant factor approximation only in the unweighted case~\cite{cc-sherali, cc-acn}.
In the weighted variant, which is more relevant in practice~\cite{shi2021scalable}, the best known approximation is logarithmic~\cite{demaine2006correlation}.
Similarly, in the case of the normalized cut objective~\cite{ncut1}, the best known algorithm rely on spectral methods which are both computationally expensive and only provide approximation guarantees which are dependent on the graph structure (and are in general not super-constant).

In this paper we study the problem of graph clustering for a broad class of objectives in which the quality of a cluster is defined based on the ratio between the number of edges in the cluster, and the total weight of vertices in the cluster. 
Here, the function assigning weights to vertices can be arbitrary.
Our goal is to maximize the sum of cluster qualities.
Since we want our algorithm
to automatically determine a good number of clusters, we use a value $\lambda$ to denote the gain (positive value of $\lambda$ or cost (negative value of $\lambda$) for each cluster. This might be viewed as a form of regularization. 
Furthermore, the relation between our approach and classical
graph clustering is analogous to the relation between 
uniform facility location and $k$-median.

We show that this objective captures and/or is closely related to multiple previously studied clustering objectives, namely normalized modularity~\cite{nmod-definition, nmod-definition3, nmod-definition2}, normalized associations and normalized cut~\cite{ncut1, ncut2}, and cluster density~\cite{BCC22}.

We show that for a broad range of parameters this objective admits a constant approximation algorithm running in linear time, which implies the first known constant-factor approximate algorithms for optimizing normalized modularity and normalized associations.

\subsection{Objective Definition}\label{sec:objective}
We use $G=(V,E)$ to denote an undirected graph with vertex set $V$ and edge set $E$, and $\deg(v)$ to denote the degree of vertex $v$.
For a set of vertices $C\subseteq V$ we define the volume of $C$ as $\vol(C) = \sum_{v\in C} \deg(v)$.
Moreover, we use $E(C)$ to denote the set of edges whose both endpoints are contained in $C$ and $E_{out}(C)$ to denote the set of edges, whose exactly one endpoint is in $C$.
Let $w : V \rightarrow \mathbb{R_{+}}$ be a function assigning weights to the vertices of $V$.
For a subset $C \subseteq V$ we will sometimes write $w(C)$ as a shorthand for $\sum_{v \in C} w(v)$.
Finally, we define \emph{clustering} of $G$ to be a partition of its vertex set $V$ into some number of disjoint and nonempty subsets, whose union is $V$.

In this paper we deal with maximizing a certain quality metric of a clustering.
Assume that $f$ is a quality metric (function) which assigns nonnegative value to any clustering.
Then, for $\alpha \in (0, 1]$ we say that a clustering $\mathcal{C}$ is $\alpha$-approximate if for any clustering $\mathcal{C}'$ we have $f(\mathcal{C}) \geq \alpha \cdot f(\mathcal{C}')$.
Similarly, we call a clustering algorithm $\alpha$-approximate if it computes an $\alpha$-approximate clustering.

\begin{definition}\label{def:quality}
Let $G=(V,E)$ be an undirected graph.
For any $C \subseteq V$ we define the quality of $C$ to be
$$
q_w(C) = \frac{2|E(C)|}{w(C)}.
$$
Moreover, assume that $\mathcal{C} = \{C_1, \ldots, C_k\}$ is a clustering of $G$ and let $\lambda \in \mathbb{R}$.
Then, we define the quality of $\mathcal{C}$ to be:

$$
Q_w^\lambda(\mathcal{C}) = \sum_{i=1}^k \left(\lambda + q_w(C_i)\right) = \sum_{i=1}^k \lambda + \frac{2|E(C_i)|}{w(C_i)}.
$$
Finally, we use $\Opt^\lambda_G$ to denote the clustering which mazimizes $Q_w^\lambda(\Opt^\lambda_G)$.
\end{definition}

We mainly focus on the case when $w(v) = \deg(v)$ for each $v \in V$, which is of particular importance, due to its connections with other clustering measures.
Because of that we will often write $Q^\lambda$ as a shorthand for $Q^\lambda_\deg$.
Observe that in this case $w(C) = \vol(C)$.
We note that since the objective is not well defined when $w(v) = 0$ for some $v \in V$, when dealing with $Q^\lambda_\deg$ we assume that the graph does not have isolated vertices.

\subsection{Our Results}
In this paper we consider the problem of finding a clustering $\mathcal{C}$ which \emph{maximizes} $Q^\lambda_w(\mathcal{C})$.
Our main technical contribution is an algorithm for optimizing $Q^\lambda_w$ up to a constant factor, whenever $\lambda \in [0,1]$.

\begin{restatable}{theorem}{thmmain}\label{thm:main-result}
For any $\lambda \in [0, 1]$ there is a algorithm which given a graph $G = (V, E, w)$ computes its clustering $\mathcal{C}$, such that $\E[Q^\lambda_w(\mathcal{C})] = \Theta(Q^\lambda_w(\Opt^\lambda_G))$.
The algorithm runs in linear time if the weight of each vertex is an integer in $[1, \mathrm{poly}(|V|)]$ and in $O(|E| \log |V|)$ time otherwise. 
\end{restatable}

The above algorithm yields constant factor approximation algorithms for several other common clustering objectives, which we describe below.
We note that in the canonical case of $w(v) = \deg(v)$, optimizing the objective is trivial whenever $\lambda \not\in[-1, 1]$ (see Lemma~\ref{lem:interpolate}).

\paragraph{Normalized modularity.}
The \emph{normalized modularity} objective was independently proposed by \cite{nmod-definition, nmod-definition3} and \cite{nmod-definition2} and then studied in~\cite{nmod-s1,nmod-s2, nmod-s3, nmod-s7, nmod-s4, nmod-s5, nmod-s6}.
It is based on the widely popular modularity objective~\cite{newman}.
The goal of the normalization is to address the commonly observed resolution limit in modularity clustering~\cite{fortunato2007resolution}.
We show that under a mild clusterability condition, we can optimize normalized modularity up to a constant factor.

Let $\ncut{}$ denote the normalized cut measure~\cite{ncut1, ncut2}.
We note that it is often assumed that a graph has a well-defined cluster structure when $\ncut(\mathcal{C}) \leq c \cdot k$ for some constant $c < 1$.
Our clusterability condition is (asymptotically) even weaker.

\begin{theorem}
There exists a constant $T$, such that the following holds.
Let $G = (V,E)$ be a graph, such that for some clustering $\mathcal{C}$ of $G$ into $k$ clusters we have $\ncut(\mathcal{C}) \leq k - T$.
Then, one can compute a constant-factor approximate clustering of $G$ with respect to the normalized modularity objective using a linear-time algorithm.
\end{theorem}

To the best of our knowledge our result constitutes the first polynomial time algorithm for maximizing normalized modularity with a provable approximation ratio.
Previous works using normalized modularity relied on a spectral approach~\cite{nmod-s4, nmod-s5}, which has a much higher running time.

\paragraph{Normalized cut and normalized associations.}  We observe that
$$
Q^0_\deg(\mathcal{C}) = \nassoc(\mathcal{C}),
$$
where $\nassoc$ is the \emph{normalized association} objective, which was introduced in the seminal paper of~\cite{ncut1} and then generalized to multiple clusters~\cite{ncut2}.
The normalized association objective is a dual of the prominent \emph{normalized cut} objective~\cite{ncut1, ncut2}, in the sense that $\ncut(\mathcal{C}) + \nassoc(\mathcal{C}) = k$, where $k$ is the number of clusters in the clustering $\mathcal{C}$.

More generally, if we denote the normalized cut by $\ncut$, we have
\begin{equation}\label{eq:ncut}
Q^\lambda_\deg(\mathcal{C}) = \sum_{i=1}^k \lambda + \frac{2|E(C_i)|}{\vol(C_i)} = \lambda \cdot k + \sum_{i=1}^k \frac{\vol(C_i) - |E_{out}(C_i)|}{\vol(C_i)} = (\lambda+1) \cdot k - \ncut(\mathcal{C}).
\end{equation}

While the values $Q^\lambda_\deg(\mathcal{C})$ and $\ncut(\mathcal{C})$ are tightly connected, there is a key difference between the typical approaches to optimizing them.
Namely, the usual approach to minimizing $\ncut$ is to fix the number of clusters $k$, and then find the clustering that optimizes the objective and has the desired number of clusters.
On the other hand, in this paper we will consider maximizing $Q^\lambda_\deg$ \emph{without} fixing the number of clusters.
Still, show that when optimizing $Q^\lambda_\deg$ the number of clusters can be tuned by varying the value of $\lambda$.

Indeed, as we show in Lemma~\ref{lem:interpolate}, for a connected graph, for $\lambda \leq -1$ the optimal clustering has only one cluster, and for $\lambda \geq 1$ the optimal clustering has the maximum possible number of clusters (each has size $1$).


Furthermore, we note that the connection between $Q^\lambda_\deg$ and $\ncut$ is not limited to theory.
It was shown empirically~\cite{ganc} that optimizing the $\nassoc$ measure greedily (which is equivalent to a certain greedy method for optimizing $Q^\lambda_\deg$) yields a very effective algorithm for optimizing normalized cut.
Namely, the algorithm delivers very similar quality to the state of the art methods, such as spectral clustering~\cite{ncut1} and Graclus~\cite{graclus}, while being much more efficient in practice.\footnote{While the analysis in~\cite{ganc} used a small number of datasets, we have implemented a greedy algorithm for minimizing $\ncut$ based on optimizing $Q^\lambda_\deg$ and tuning $\lambda$.
Having evaluated the method on a much larger amount of datasets, we were able to observe that our greedy approach typically outperforms spectral clustering. We will provide more details in the full version of this paper.}
Hence, we believe that the connection to the normalized cut is relevant both in theory and practice.

\paragraph{Edge density.}
Given $C \subseteq V$, the \emph{edge density} of $C$ is defined as $d(C) := |E(C)| / |C|$. 
If we set $w(v) = 1$ for each $v \in V$, we have
$$
Q^0_w(\mathcal{C}) = \sum_{i=1}^k \frac{2|E(C_i)|}{|C_i|} = 2 \cdot \sum_{i=1}^k d(C_i).
$$
The problem of
maximizing sum of cluster densities has been studied in 
\cite{BCC22,AGGMT15}. While a $2$-approximation algorithm for the problem is
known under a different name \cite{AGGMT15} (see also \cite{BCC22}),
Theorem~\ref{thm:main-result} implies a different a constant approximation algorithm for the problem using a more general approach.

\section{Basic Properties of Our Objective}
We show that tuning the parameter $\lambda$ allows us to interpolate between two extreme clusterings: one that has a single cluster and one that has the maximum possible number of clusters.

\begin{lemma}\label{lem:interpol}
Let $G=(V,E,w)$ be a connected graph and let $M = \frac{2|E|}{\min_{v\in V} w(v)}$.
Then, (a) for any $\lambda < -M$ there is a unique clustering maximizing $Q^{\lambda}_w$, which has exactly one cluster containing all vertices of $G$, and (b) for any $\lambda > M$ there is a unique clustering maximizing $Q^\lambda_w$, which consists solely of clusters of size $1$.
\end{lemma}

\begin{proof}
For any clustering $\mathcal{C} = \{C_1, \ldots, C_k\}$ consisting of $k$ clusters we have $$Q^0_w(\mathcal{C}) = \sum_{i=1}^k \frac{2|E(C_i)|}{w(C_i)} \leq \sum_{xy \in E} \frac{2}{\min_{v\in V} w(v)} = M.$$
Thus $Q^0_\deg(\mathcal{C}) \in [0, M].$

By definition, we also have that $Q^\lambda_w(\mathcal{C}) = \lambda \cdot k + Q^0_w(\mathcal{C})$.
By looking at this formula it is easy to see that for $\lambda > M$ if a clustering has any cluster of size more than $1$, splitting this cluster arbitrarily increases the objective.
Namely, the $\lambda \cdot k$ summand increases by more than $M$, which surely dominates any resulting change to $Q^0_w(\mathcal{C})$.
As a result, the optimal clustering has the maximum possible number of clusters.
Similarly, for $\lambda < -M$ the optimal clustering has only one cluster.
\end{proof}

Moreover, in the case of $w(v) = \deg(v)$ we can effectively limit $\lambda$ to $[-1, 1]$.

\begin{lemma}\label{lem:interpolate}
Let $G$ be a connected graph with at least $3$ vertices.
Then, (a) for any $\lambda \leq -1$ there is a unique clustering maximizing $Q^{\lambda}$, which has exactly one cluster containing all vertices of $G$, and (b) for any $\lambda \geq 1$ there is a unique clustering maximizing $Q^1$, which consists solely of clusters of size $1$.
\end{lemma}

\begin{proof}
We have $Q^{-1}(\mathcal{C}) = -\ncut(\mathcal{C})$.
Hence, our goal is to find a clustering which minimizes $\ncut$ (without fixing the number of clusters).
It is easy to see that in the optimal clustering there can be no inter-cluster edges, which for a connected graph is the case only when all vertices are in the same cluster.
This argument can be easily extended to $\lambda < -1$.

Let us now consider the case when $\lambda = 1$ and denote by $n$ the number of vertices of $G$ (again, the argument easily extends to the case of $\lambda > 1$).
We first show that a clustering into a large number of clusters must contain some number of clusters of size $1$.
Namely, if we cluster $n$ vertices into $k$ clusters and $s$ clusters have size $1$, then the remaining $k-s$ clusters have size at least $2$, and so we must have $s + 2(k-s) \leq n$, which implies $s \geq 2k - n$.
This implies that if $\mathcal{C}$ is a clustering into $k$ clusters and $k < n$, we have $\ncut(\mathcal{C}) > 2k-n$.
Indeed, in the normalized cut measure, the contribution of each cluster $C$ is $|E_{out}(C)| / \vol(C)$ (see Equation~\ref{eq:ncut}), which is equal to $1$ when $|C|=1$ (here we need the assumption that the graph has at least $2$ vertices).
Hence, each cluster of size $1$ contributes $1$ to the normalized cut.
Additionally, we have at least one cluster of size $2$, which makes the inequaltiy strict.

Let us now consider three possible cases for the number of clusters.
If $k = n$ (all clusters have size $1$) we have
$$
Q^1(\mathcal{C}) = 2n - \ncut(\mathcal{C}) = n.
$$
If $1 < k < n$ we have
$$
Q^1(\mathcal{C}) = 2k - \ncut(\mathcal{C}) < 2k - 2k + n = n.
$$
Finally, for $k=1$, since $n \geq 3$
$$
Q^1(\mathcal{C}) = 2 < n.
$$
Out of these two options, the value of $Q^1$ is maximized only when each vertex is in a separate cluster.
\end{proof}

As an interesting side result, we note that we can compute a nontrivial upper bound on $Q^0$ by computing a maximum spanning tree of a properly weighted graph.

\begin{restatable}{theorem}{thmub}\label{thm:ub}
  Let $G=(V,E)$ be a connected graph. For each edge $e=(u,v)$ define its weight $W(e) = \frac{1}{\max\{deg(u),\deg(v)\}}$. Let $M$ be the cost of the maximum spanning
  forest of the resulting graph. Then,
  $
  \frac{M}{3\sqrt{n}}-\frac{1}{3} \le  Q^0(\Opt_G)\le 2M,
  $
  where $\Opt_G$ a clustering of $G$ which maximizes $Q^0(\Opt_G)$.
\end{restatable}

\section{Constant-factor Approximation Algorithms}
\label{sec:alg}
In this section we present our approximation algorithms.
Specifically, we first give a constant-factor approximation algorithm for optimizing $Q^0_w$, and then show how to extend it to a constant-factor approximation algorithm for $Q^\lambda_\deg$, as long as $\lambda \in (0, 1]$.
We also present approximation algorithm for the normalized modularity measure in Section~\ref{sec:nmod}.
We note that as long as node weights are integers, the algorithms work in linear time (and in near-linear time otherwise).

Let us now focus on the case when $\lambda=0$.
Our algorithm first transforms the input graph in several steps into an input to the following 
discrete optimization problem, which we call \emph{capacitated vertex-weighted assignment problem (CVWAP)}.
During our transformation we make sure that in expectation a constant-factor approximate solution to the CVWAP problem in the transformed graph corresponds to a constant factor approximation of the $Q^0_w$ objective in the original graph.

The input to CVWAP is a bipartite graph $G = ((S,T), E, w)$ with positive vertex weights given by $w$.
The goal is to find a clustering $\mathcal{C} =\{ C_1, \ldots, C_k\}$ of a \emph{subset} of $S \cup T$, such that for 
each $i$, $|C_i \cap S| = 1$ and $w(C_i \cap T) \leq 2w(C_i \cap S)$. The objective is to maximize $v(\mathcal{C}) := \sum_{i=1}^k \frac{2|C_i \cap T|}{3w(C_i \cap S)}$.
We then prove that a constant-approximate solution to the CVWAP can be computed in linear time using a simple greedy algorithm, similar to the greedy algorithm for maximum weight bipartite matching. 

We define a \emph{partial clustering}
of $G=(V,E)$ to be a partition of some subset $U \subseteq V$.
Let $\mathcal{C} = \{C_1, \ldots, C_k\}$ be a partial clustering of $G$.
We extend the definition of $Q^0_w$ to a partial clustering $\mathcal{C}$ in the natural way by defining 
$
Q^0_w(\mathcal C) = \sum_{C \in \mathcal{C}} q_w(C).
$
Since each partial clustering $\mathcal{C}$ can be extended to a (non-partial) clustering $\mathcal{C}'$, such that $Q^0_w(\mathcal{C}) \leq Q^0_w(\mathcal{C}')$ without loss of generality in the following we can turn our attention to finding an $\alpha$-approximate \emph{partial} clustering of $G$.

We now show that (losing a factor of $1/4$ in the expected approximation ratio) we can reduce the input graph to a bipartite graph, in which all edges have their higher-weight endpoint on the same side, and solve the CVWAP problem on this bipartite graph. 

Our transformation to a bipartite graph simply colors each vertex randomly either red or blue (both with probability $1/2$), then selects one of the colors to be $S$ (one side of the bipartition) and the other to be $T$ (the other side). Then it keeps all edges $(s,t)$ with $s\in S$ to $t\in T$ that have $w(s) \geq w(t)$ and discards the rest. The next lemma shows that this construction
in expectation preserves the quality of an optimal solution within a factor of $1/4$.
\begin{restatable}{restatablelemma}{lembipartite}\label{lem:bipartite}
Let $G = (V, E, w)$ be a graph.
There exists a randomized linear-time algorithm that outputs a bipartite graph $H = ((S \cup T), E_H, w)$, such that:
		(i) $H$ is a subgraph of $G$, (ii) $S \cup T = V$,
		(iii) for each edge $(s,t) \in E_H$, where $s \in S$ and $t \in T$, we have $w(s) \geq w(t)$, and
		(iv) $\E[Q^0_w(\Opt_H)] \geq Q^0_w(\Opt_G)/4$.
	where $\Opt_H$ and $\Opt_G$ are optimal clusterings in $H$ and $G$, respectively.
\end{restatable}

Our next step is to argue that solving the CVWAP problem on the transformed instance gives a constant approximation algorithm.
For this purpose, we first show that we can restrict ourselves to considering solutions where each cluster has a 
single vertex on the $S$-side and a bounded total weight on the $T$-side. 
In this step we pay a factor of $1/7$ in the approximation ratio. 

\begin{restatable}{restatablelemma}{lemstars}\label{lem:stars}
Let $H = ((S,T), E, w)$ be a bipartite graph, where for each $(s,t) \in E$, such that $s \in S$ and $t \in T$ we have $w(s) \geq w(t)$.
Then, there exists a partial clustering $\mathcal{C} = \{C_1, \ldots, C_k\}$ of $S \cup T$, such that the following hold:
(i) $|C_i \cap S| = 1$ for each $i$,
(ii) $w(C_i \cap T) \leq 2w(C_i \cap S)$ for each $i$, and
(iii) $Q^0_w(\mathcal{C}) \geq Q^0_w(\Opt^0_H)/7$.
\end{restatable}

Before we prove Lemma~\ref{lem:stars}, we need to show some other properties first.
We start with an auxiliary lemma, which shows that for the special case of a connected bipartite graph with only one vertex in $S$, we can greedily construct a good cluster who does not have too much weight on the $T$-side.
\begin{lemma}\label{lem:choosec}
Let $G = ((S \cup T), E, w)$ be a connected bipartite graph, where $S = \{s\}$ and $w(s) \geq w(t)$ for each $t \in T$.
Let $t_1, \ldots, t_l$ be vertices of $T$ sorted such that $w(t_1) \leq w(t_2) \leq \ldots \leq w(t_l)$ and let $1 \leq j \leq l$ be the maximum index such that $\sum_{i=1}^j w(t_i) \leq 2w(s)$.
Let $\Opt_G$ be a subset of $V$ that maximizes $Q^0_w(\Opt)$ and let $C = \{s, t_1, \ldots, t_j\}$.
Then, (i) $3 q_w(C) \geq \frac{2|C\cap T|}{w(s)}$, and (ii) $3 q_w(C) \geq Q^0_w(\Opt_G)$.
\end{lemma}

\begin{proof}
Let $C$ be as defined in the lemma.
We observe that $\Opt_G$ must contain $s$ as otherwise it does not contain edges
and $Q^0_w(\Opt_G)$ would be $0$. Since $G$ is connected, for each vertex $v \in T$ 
there is an edge $(s,v)$. Thus, we have
\[
	3 q_w(C) = \frac{6|C \cap T|}{w(s) + \sum_{i=1}^{|C \cap T|} w(t_i)} \geq \frac{2|C \cap T|}{w(s)}.
\]
We observe that $\Opt_G \cap T$ contains the $|\Opt_G \cap T|$ vertices of $T$ with smallest weights as otherwise, we could improve the solution by swapping a vertex in $\Opt_G$ with 
a vertex of smaller weight. If $|\Opt_G \cap T| \leq |C \cap T| = j$, we have
\[
	Q^0_w(\Opt_G) = \frac{2|\Opt_G \cap T|}{w(s) + \sum_{i=1}^{|\Opt_G \cap T|} w(t_i)} \leq \frac{2|C \cap T|}{w(s)} \leq 3 q_w(C).
\]

Let us now assume the opposite case, that is $|\Opt_G \cap T| > |C \cap T| = j$.
To complete the proof, we use two properties.
First, observe that $\frac{x}{\sum_{i=1}^x w(t_i)}$ is nonincreasing as a function of $x$, since $w(t_i)$ are sorted in nondecreasing order.
Second, the assumption that $|\Opt_G \cap T| > |C \cap T|$ implies that $\sum_{i=1}^j w(t_i) \geq w(s)$.
Indeed, if this was not the case, we would be able to increase $j$ without violating the constraint that $\sum_{i=1}^j w(t_i) \leq 2w(s)$ (note that $w(s) \geq w(t_i)$ for all $i$).
We use these two facts to derive
\begin{align*}
Q^0_w(\Opt_G) & =
	\frac{2|\Opt_G \cap T|}{w(s) + \sum_{i=1}^{|\Opt_G \cap T|} w(t_i)} \leq \frac{2|\Opt_G \cap T|}{\sum_{i=1}^{|\Opt_G \cap T|} w(t_i)} \leq \frac{2j}{\sum_{i=1}^j w(t_i)} \leq \frac{2|C \cap T|}{w(s)} \leq 3 q_w(C).
\end{align*}
Note that we use the properties in third and second to last step respectively.
\end{proof}

By applying Lemma~\ref{lem:choosec} we get the following.

\begin{corollary}\label{cor:boundweight}
Let $G = ((S,T), E, w)$ be a bipartite graph, where for each $(s,t) \in E$, such that $s \in S$ and $t \in T$ we have $w(s) \geq w(t)$.
Consider $C \subseteq V$, such that $|C \cap S| = 1$.
Then, there exists $C' \subseteq C$ such that (i) $|C' \cap S| = 1$, (ii) $w(C' \cap T) \leq 2w(C' \cap S)$, and (iii) $2 q_w(C') \geq q_w(C) - q_w(C')$.
\end{corollary}

\begin{proof}
Let $S'=\{s\}$ and $T'$ be the neighbors of $s$ that are in $C$.
Then the subgraph $H$ induced by $S' \cup T'$ is a connected bipartite graph 
with $S'=\{s\}$, i.e. it satisfies the prerequisites of Lemma \ref{lem:choosec}.
Thus, we can apply Lemma \ref{lem:choosec} and we obtain a cluster $C'$
that has $|C'\cap S| = 1$, $w(C' \cap T) \leq 2w(C' \cap S)$, and
$3q_w(C')\ge q_w(\Opt_H) \ge q_w(C)$, which implies $2q_w(C') \ge q_w(C) - q_w(C')$.
\end{proof}

We now show the first reduction that allows us to replace a cluster that has multiple vertices on the $S$-side with two clusters.

\begin{lemma}\label{lem:splitcluster}
Let $G = ((S,T), E, w)$ be a bipartite graph, where for each $(s,t) \in E$, such that $s \in S$ and $t \in T$ we have $w(s) \geq w(t)$.
Consider $C \subseteq V$, such that $|C \cap S| > 1$.
Then, there exists a partial clustering of $C$ into $C_1$ and $C_2$, such that $|C_1 \cap S| = 1$, $w(C_1 \cap T) \leq 2w(C_1 \cap S)$, and $6 q_w(C_1) \geq q_w(C) - (q_w(C_1) + q_w(C_2))$.
\end{lemma}

\begin{proof}
Let $p$ be the element of $C \cap S$ that has the smallest weight (an arbitrary one, in case of a tie).
Now choose $C_1 \subseteq C$ just like we chose $C$ in the statement of Lemma~\ref{lem:choosec} (here, we consider the subgraph of $G$ induced on $C$). 
We set $C_1 = \{p, t_1, \ldots, t_j\}$ and $C_2 = C \setminus C_1$.
We observe that $C_1$ satisfies $|C_1 \cap S| = 1$ and
$w(C_1 \cap T) \leq 2w(C_1 \cap S)$ by construction.

To prove that $C_1$ and $C_2$ satisfy the claim, we first upper bound $q_w(C) - (q_w(C_1) + q_w(C_2))$.
(Note that we do not yet use the properties of $C_1$ and $C_2$, other than the fact that their weights are smaller than the weight of $C$.)
	\begin{align*}
		q_w(C) - & (q_w(C_1)  + q_w(C_2)) = \sum\limits_{\substack{(s,t) \in E \\ s \in (C \cap S) \\ t \in (C \cap T)}}  \mathbbm{1}_{s \in C_1, t \in C_1} \left( \frac{2}{w(C)} - \frac{2}{w(C_1)}\right)\\
	& \qquad +   \mathbbm{1}_{s \in C_2, t \in C_2} \left( \frac{2}{w(C)} - \frac{2}{w(C_2)}\right) +   \mathbbm{1}_{s \in C_1, t \in C_2} \frac{2}{w(C)} + \mathbbm{1}_{s \in C_2, t \in C_1} \frac{2}{w(C)} \\
		& \leq \sum\limits_{\substack{st \in E \\ s \in (C \cap S) \\ t \in (C \cap T)}} \mathbbm{1}_{s \in C_1, t \in C_2} \frac{2}{w(C)} + \mathbbm{1}_{s \in C_2, t \in C_1} \frac{2}{w(C)}
	\end{align*}

We first bound both summands separately.
Since $C_1$ is the optimal subset such that $C_1 \cap S = p$, we have
	\begin{align*}
		\sum \limits_{\substack{(s,t) \in E \\ s \in (C \cap S) \\ t \in (C \cap T)}}  \mathbbm{1}_{s \in C_1, t \in C_2} \frac{2}{w(C)}  & = \sum\limits_{\substack{(s,t) \in E \\ s = p \\ t \in (C_2 \cap T)}}  \frac{2}{w(C)}
		\leq  \sum\limits_{\substack{(s,t) \in E \\ s = p \\ t \in (C_2 \cap T)}}  \frac{2}{w((C_2 \cap T)\cup \{p\})}\\
		& = q_w((C_2 \cap T) \cup \{p\}) \leq 3 q_w(C_1).
	\end{align*}
Note that in the last step we use the second claim of Lemma~\ref{lem:choosec}.

To upper bound the second summand we use the fact that $p$ is the lowest-weight vertex of $C$ on the $S$-side.
Note that this implies $w(C) \geq |C \cap S|w(p) \geq |C_2 \cap S|w(p)$.
Hence,
	\begin{align*}
		\sum\limits_{\substack{(s,t) \in E \\ s \in (C \cap S) \\ t \in (C \cap T)}} & \mathbbm{1}_{\substack{s \in C_2\\ t \in C_1}} \frac{2}{w(C)} \leq \frac{2 |C_2 \cap S|\cdot|C_1 \cap T|}{w(C)} \leq \frac{2 |C_2 \cap S|\cdot|C_1 \cap T|}{w(p) |C_2 \cap S|} = \frac{2 \cdot|C_1 \cap T|}{w(p)} \leq 3\cquality(C_1).
	\end{align*}
Hence, we get that $q_w(C) - (q_w(C_1) + q_w(C_2)) \leq 6 q_w(C_1)$.
\end{proof}

We can now prove Lemma~\ref{lem:stars}.

\begin{proof}[Proof of Lemma~\ref{lem:stars}]
Consider the following process, which produces a sequence of partial clusterings of $H$. 
Start with $\Opt_H^0$.
As long as there is a cluster that does not satisfy conditions (i) or (ii) of the lemma, we apply either Corollary~\ref{cor:boundweight} (if the cluster only violates condition (ii)) or Lemma~\ref{lem:splitcluster} (otherwise).
Assume the process consists of $r$ steps.
As a result we obtain a sequence of clusterings $\mathcal{C}_1, \ldots, \mathcal{C}_{r+1}$, where $\mathcal{C}_1 = \Opt_H^0$, and $C_{r_1}$ satisfies both (i) and (ii).
To complete the proof, we show that $\mathcal{C}_{r+1}$ satisfies (iii).

Consider step $i$ (i.e. when we create $\mathcal{C}_{i+1}$).
Note that both Lemma~\ref{lem:splitcluster} and Lemma~\ref{cor:boundweight} produce one cluster that satisfies both conditions.
Regardless of which of the two reductions is used in the step, denote this cluster by $D_i$.
Observe that $\mathcal{C}_{i+1}$ is obtained from $\mathcal{C}_i$ by either splitting a cluster into two clusters or by replacing a single cluster.
By using Lemma~\ref{lem:splitcluster} and Corollary~\ref{cor:boundweight} we have $Q^0_w(\mathcal{C}_i) - Q^0_w(\mathcal{C}_{i+1}) \leq 6 q_w(D_i)$.
Moreover, each $D_i$ already satisfies both conditions and is not going to be replaced or split further, and thus is an element of $\mathcal{C}_{r+1}$.
This implies that $\sum_{i=1}^{r} q_w(D_i) \leq Q^0_w(\mathcal{C}_{r+1})$.

Hence, we get that $Q^0_w(\mathcal{C}_1) - Q^0_w(\mathcal{C}_{r+1}) = \sum_{i=1}^{r} (Q^0_w(\mathcal{C}_i) - Q^0_w(\mathcal{C}_{i+1})) \leq \sum_{i=1}^{r} 6 q_w(D_i) \leq 6\cdot Q^0_w(\mathcal{C}_{r+1}).$
By rearranging terms, we have $Q^0_w(\mathcal{C}_{r+1}) \geq Q^0_w(\mathcal{C}_1)/7 = Q^0_w(\Opt_H^0)/7$.
\end{proof}

Using Lemma~\ref{lem:stars} we can obtain the final reduction to the CVWAP problem.

\begin{restatable}{restatablelemma}{lemapxr}\label{lem:apxr}
Assume there exists an $\alpha$-approximate algorithm for CVWAP.
Then, there exists a polynomial time randomized algorithm that, 
given an unweighted graph $G$, finds a partial clustering $\mathcal{C}$ of $V$, such that 
$\E[Q^0_w(\mathcal{C})] \geq \frac{\alpha}{84} \cdot Q^0_w(\Opt_G)$,
where $\Opt_G$ denotes an optimal clustering
in $G$.
\end{restatable}

\subsection{Solving the CVWAP Problem}
We now describe a greedy $1/2$-approximate algorithm for the CVWAP problem. The algorithm and its analysis are analogous to the greedy algorithm for maximum weight matching.

Let $G = ((S,T), E, w)$ be the input to the CVWAP problem.
Observe that the solution to the CVWAP problem can be uniquely described by a set of intra-cluster edges, that is edges whose both endpoints belong to the same element of the vertex partition.
The condition that each element of the partial clustering can contain at most one element of $S$ translates to the condition that each element of $T$ may have at most one incident edge in the solution.
In the following, we represent clusters by its edges assuming that the vertex set of a cluster
is defined by the union of the vertices of its edges.

The algorithm first assigns each edge $(s,t)$ ($s \in S, t\in T$) a weight of $1/(w(s) + w(t))$ and then sorts the edges in non-increasing order of weights.
We begin with a solution where we have a singleton cluster for each element of $S$.
Then, the algorithm considers the edges in sorted order.
For an edge $(s,t)$ it adds vertex $t$ to the cluster of $s$, if only this does not violate the capacity constraints and no edge incident to $t$ has been previously added to the solution.

\begin{restatable}{restatablelemma}{cvwapapx}
The above algorithm computes a $1/2$-approximate solution to the CVWAP problem.
\end{restatable}

\begin{proof}
Denote by $M_i$ the partial solution after adding $i$ edges.
Hence, we have $S_0 = \emptyset$, $M_1 \subseteq M_2 \subseteq \ldots \subseteq M_q$, where $q$ is the number of edges in the final solution.
For the proof, we will define a sequence $O_0, \ldots, O_q$, where $O_0$ is the optimal solution of the CVWAP problem, and $O_0 \supseteq O_1 \supseteq \ldots \supseteq O_q = \emptyset$.
Moreover, we will maintain the invariant that $M_i \cup O_i$ is a valid solution.

Consider step $i$, in which we obtain $M_{i+1} = M_i \cup \{(s,t)\}$.
We construct $O_{i+1}$ as follows.
$O_{i+1}$ contains all edges of $O_i$ with possibly two exceptions.
First, if there is an edge with an endpoint $t$ in $O_i$ (clearly, there can be at most one), exclude it from $O_{i+1}$.
Second, if there is an edge with an endpoint $s$ in $O_i$, exclude the edge with the lowest weight among all such edges.
Hence, we have $|O_{i+1}| \geq |O_i| - 2$.

We now show that $M_i \cup O_i$ is a valid solution by using induction.
Clearly, the claim holds for $i=0$.
For an inductive step, we see immediately that at most one edge is incident to any vertex of $T$ thanks to the construction.
Moreover, the solution does not violate capacity constraints on the vertices of $S$ for the following reason.
When we add $st$ to the solution, either $O_i$ has no edges incident to $s$ (and $M_i \cup S_i$ does not violate the capacity constraint, since $M_i$ does not violate it), or we remove an edge $(s,t')$ of lower weight from $S_i$.
The weight of $(s,t')$ is not greater than the weight of $(s,t)$ (thanks to the order in which we consider edges) which implies that $w(t') \geq w(t)$.
Hence, we free up enough capacity to add the edge without violating the constraint.

Finally, we have that $O_q = \emptyset$, as otherwise we could extend $M_q$ with at least one more edge, which is clearly impossible.
To show the approximation ratio we prove that $v(M_{i}) - v(M_i-1) \geq (v(O_{i-1}) - v(O_{i}))/2$.
This follows from the fact that thanks to the greedy order of iteration, the contribution of the (at most two) edges in $O_{i-1} \setminus O_i$ to $v(\cdot)$ is at most the contribution of the only edge of $M_{i+1} \setminus M_i$.

As a result, we get 
\begin{align*}
	v(M_q) & = v(M_q) - v(M_0) = \sum_{i=1}^q (v(M_i) - v(M_{i-1}))\\
	& \geq 1/2 \sum_{i=1}^q (v(O_{i-1}) - v(O_i)) = 1/2 (v(O_0) - v(O_q)) = v(O_0)/2.
\end{align*}

\end{proof}

Finally, let us show how to extend our algorithm to the case when $\lambda \in [0, 1]$.

\begin{lemma}\label{lem:gt}
Let $\lambda > 0$.
An $\alpha$-approximate algorithm for optimizing $Q^0_w$ implies an $\alpha/(1+\alpha)$-approximate algorithm for optimizing $Q^\lambda_w$.
\end{lemma}

\begin{proof}
Let $\mathcal{C}$ be a clustering.
We have $Q^\lambda_w(\mathcal{C}) = \lambda \cdot k + Q^1_w(\mathcal{C}).$
We consider two cases: the first case is when $\lambda \cdot k$ is at least a constant factor of $Q^\lambda(\mathcal{C})$, and the second case is the remaining one.

Specifically, in the first case we assume $\lambda \cdot k > \alpha/(1+\alpha) \cdot Q^\lambda(\mathcal{C})$.
In this case outputting a clustering in which each node is in a separate cluster achieves objective of $\lambda \cdot n$, which using the above inequality provides $\alpha / (1+\alpha)$ approximation.

In the remaining case we have $\lambda \cdot k \leq \alpha/(1+\alpha) \cdot Q^\lambda(\mathcal{C})$, which implies $Q^1(\mathcal{C}) = Q^\lambda(\mathcal{C}) - \lambda \cdot k \geq (1 - \alpha / (1+\alpha)) \cdot Q^\lambda(\mathcal{C}) = 1 / (1+\alpha) \cdot Q^\lambda(\mathcal{C})$.
By multiplying both sides of this inequality by $\alpha$ we get the desired.
\end{proof}

We now combine the reduction of Lemma~\ref{lem:apxr} with the approximation algorithm for the CVWAP problem (and, whenever $\lambda \in (0, 1]$ with the reduction of Lemma~\ref{lem:gt}) to obtain our main theoretical result.
We note that the approximation ratio of the algorithm is $1/169$.

\thmmain*

\begin{proof}
Let us first consider $\lambda = 0$.
The algorithm first transforms the input graph to an instance of the CVWAP problem.
By Lemma~\ref{lem:bipartite}, this can be done in linear time.
The first step of solving CVWAP problem consists in sorting the edges in nonincreasing order, where the sorting key assigned to an edge $uv$ is $1/(w(u) + w(v))$.
This is equivalent to using $w(u) + w(v)$ as a sorting key and sorting in nondecreasing order.
The sorting takes $O(|E| \log |V|)$ time, or only linear time if the sorting keys are small integers, in which case we can use radix-sorting.

It is easy to see that the remaining steps of the algorithm for the CVWAP problem can be easily implemented in linear time.
Since the algorithm for the CVWAP problem is 1/2-approximate, by Lemma~\ref{lem:apxr} we get that the expected approximation ratio of the entire algorithm for $\lambda = 0$ is $\alpha = 1/168$.

To handle the case of $\lambda \in (0, 1]$ we apply Lemma~\ref{lem:gt}, and observe that the proof of the lemma is constructive and uses a linear-time algorithm.
We have that the approximation ratio is $\alpha / (1+\alpha) = 1/169$.
\end{proof}

\subsection{Normalized Modularity}\label{sec:nmod}
For a cluster $C \subseteq V$ in a graph with $m$ edges, the (unnormalized) modularity~\cite{newman} of $C$ is defined as
$$
M(C) = \frac{1}{m}\left(|E(C)| - \frac{\vol^2(C)}{m}\right).
$$
This can be extended to the modularity of a clustering $\mathcal{C} = \{C_1, \ldots, C_k\}$ in a natural way:
$$
\mmod(\mathcal{C}) = \sum_{i=1}^k M(C_i).
$$
The normalized modularity $\nmod$ is obtained by normalizing the contribution of each cluster by its volume. As observed in~\cite{nmod-definition, nmod-definition2, nmod-definition3, ganc}, we have
\begin{align}\label{eq:nmod}
\nmod(\mathcal{C}) & = \sum_{i=1}^k \frac{M(C_i)}{\vol(C_i)} = \frac{1}{m}\left( \sum_{i=1}^k \frac{|E(C_i)|}{\vol(C_i)} - \sum_{i=1}^k\frac{\vol(C_i)}{m}\right)
 = \frac{1}{m}\left(\frac{Q^0_\deg(\mathcal{C})}{2} - 1\right).
\end{align}

\begin{lemma}
Let $G$ be a graph and let $\mathcal{C}$ be an $\alpha$-approximate clustering of $G$ with respect to $Q^0$.
Moreover, assume that for some $\epsilon > 0$ there exists a clustering $\mathcal{C}'$ of $G$ into $k$ clusters such that $\ncut(\mathcal{C}') \leq k - (1+\epsilon)/\alpha$.
Then $\nmod(\mathcal{C}) \geq \frac{\epsilon \cdot \alpha}{1+\epsilon} \nmod(\Opt_{\nmod})$ where $\Opt_{\nmod}$ is a  clustering of $G$ of maximum normalized modularity.
\end{lemma}

\begin{proof}
Thanks to Equation~\ref{eq:nmod} we observe that $\Opt_{\nmod}$ is also an optimal clustering with respect to the $Q^0_\deg$ measure.
This implies
\[
Q^0_\deg(\mathcal{C}) \geq \alpha \cdot Q^0_\deg(\Opt_{NMod}) \geq \alpha \cdot Q^0_\deg(\mathcal{C}') = \alpha (k - \ncut(\mathcal{C}')) \geq \alpha (1+\epsilon) / \alpha = 1+\epsilon
\]

By combining the above with fact that $(x-1)/x$ is increasing in $x$ (for $x \geq 1$) we obtain
\[
\frac{Q^0_\deg(\mathcal{C}) - 1}{Q^0_\deg(\mathcal{C})} \geq \frac{1+\epsilon-1}{1+\epsilon} = \frac{\epsilon}{1+\epsilon},
\]
which can be rewritten as $Q^0_\deg(\mathcal{C}) - 1 \geq \frac{\epsilon}{1+\epsilon}Q^0_\deg(\mathcal{C})$. Hence
\[
\frac{\nmod(\mathcal{C})}{\nmod(\Opt_{\nmod})} = \frac{Q^0_\deg(\mathcal{C}) - 1}{Q^0_\deg(\Opt_{\nmod})-1} \geq \frac{\epsilon \cdot Q^0_\deg(\mathcal{C})}{(1+\epsilon)Q^0_\deg(\Opt_{\nmod})} \geq \frac{\epsilon \cdot \alpha}{1+\epsilon}.
\]
\end{proof}




\section{Conclusion}

In this paper, we have studied a general clustering objective which is connected to well known clustering objectives such as modularity, normalized cut and cluster density and
where the quality of a cluster is determined by the number of edges inside the cluster divided by the sum of its vertex weights. Our goal is to maximize the sum of cluster qualities where we regularize by $\lambda$ times the number of clusters. For a wide range of choice of $\lambda$ we develop a $O(1)$-approximation under this quality measure. 

\bibliographystyle{alpha}
\bibliography{references}

\newcommand{\etalchar}[1]{$^{#1}$}
\begin{thebibliography}{AGG{\etalchar{+}}15}

\bibitem[ACN08]{cc-acn}
Nir Ailon, Moses Charikar, and Alantha Newman.
\newblock Aggregating inconsistent information: Ranking and clustering.
\newblock {\em J. ACM}, 55(5), nov 2008.

\bibitem[AGG{\etalchar{+}}15]{AGGMT15}
Haris Aziz, Serge Gaspers, Joachim Gudmundsson, Juli{\'a}n Mestre, and Hanjo
  Taubig.
\newblock Welfare maximization in fractional hedonic games.
\newblock In {\em Twenty-Fourth International Joint Conference on Artificial
  Intelligence}, 2015.

\bibitem[BBC{\etalchar{+}}13]{nmod-s2}
Marianna Bolla, Brian Bullins, Sorathan Chaturapruek, Shiwen Chen, and Katalin
  Friedl.
\newblock When the largest eigenvalue of the modularity and normalized
  modularity matrix is zero.
\newblock {\em arXiv preprint arXiv:1305.2147}, 2013.

\bibitem[BBC{\etalchar{+}}15]{nmod-s1}
Marianna Bolla, Brian Bullins, Sorathan Chaturapruek, Shiwen Chen, and Katalin
  Friedl.
\newblock Spectral properties of modularity matrices.
\newblock {\em Linear Algebra and Its Applications}, 473:359--376, 2015.

\bibitem[BCC22]{BCC22}
Cristina Bazgan, Katrin Casel, and Pierre Cazals.
\newblock Dense graph partitioning on sparse and dense graphs.
\newblock In {\em 18th Scandinavian Symposium and Workshops on Algorithm
  Theory, {SWAT} 2022}, volume 227 of {\em LIPIcs}, pages 13:1--13:15, 2022.

\bibitem[Bol11]{nmod-definition}
Marianna Bolla.
\newblock Penalized versions of the newman-girvan modularity and their relation
  to normalized cuts and k-means clustering.
\newblock {\em Physical review E}, 84(1):016108, 2011.

\bibitem[Bol14]{nmod-s3}
Marianna Bolla.
\newblock Modularity spectra, eigen-subspaces, and structure of weighted
  graphs.
\newblock {\em European Journal of Combinatorics}, 35:105--116, 2014.
\newblock Selected Papers of EuroComb'11.

\bibitem[CALN22]{cc-sherali}
Vincent Cohen-Addad, Euiwoong Lee, and Alantha Newman.
\newblock Correlation clustering with sherali-adams, 2022.

\bibitem[DEFI06]{demaine2006correlation}
Erik~D Demaine, Dotan Emanuel, Amos Fiat, and Nicole Immorlica.
\newblock Correlation clustering in general weighted graphs.
\newblock {\em Theoretical Computer Science}, 361(2-3):172--187, 2006.

\bibitem[DGK07]{graclus}
Inderjit~S. Dhillon, Yuqiang Guan, and Brian Kulis.
\newblock Weighted graph cuts without eigenvectors {A} multilevel approach.
\newblock {\em {IEEE} Trans. Pattern Anal. Mach. Intell.}, 29(11):1944--1957,
  2007.

\bibitem[DLT15]{apx-mod}
Thang~N Dinh, Xiang Li, and My~T Thai.
\newblock Network clustering via maximizing modularity: Approximation
  algorithms and theoretical limits.
\newblock In {\em 2015 IEEE International Conference on Data Mining}, pages
  101--110. IEEE, 2015.

\bibitem[FB07]{fortunato2007resolution}
Santo Fortunato and Marc Barthelemy.
\newblock Resolution limit in community detection.
\newblock {\em Proceedings of the national academy of sciences}, 104(1):36--41,
  2007.

\bibitem[JM15]{nmod-s7}
Hansi Jiang and Carl Meyer.
\newblock Relations between adjacency and modularity graph partitioning.
\newblock {\em arXiv preprint arXiv:1505.03481}, 2015.

\bibitem[LZ21]{nmod-s6}
Leon Lan and Alessandro Zocca.
\newblock Refining bridge-block decompositions through two-stage and recursive
  tree partitioning.
\newblock {\em arXiv preprint arXiv:2110.06998}, 2021.

\bibitem[New06]{newman}
M.~E.~J. Newman.
\newblock Modularity and community structure in networks.
\newblock {\em PNAS}, 103(23):8577--8582, 2006.

\bibitem[NO18a]{nmod-s5}
L{\'a}szl{\'o} Nagy and Mih{\'a}ly Ormos.
\newblock Friendship of stock market indices: A cluster-based investigation of
  stock markets.
\newblock {\em Journal of Risk and Financial Management}, 11(4):88, 2018.

\bibitem[NO18b]{nmod-s4}
L{\'a}szl{\'o} Nagy and Mihaly Ormos.
\newblock Review of global industry classification.
\newblock In {\em ECMS}, pages 66--73, 2018.

\bibitem[SDE{\etalchar{+}}21]{shi2021scalable}
Jessica Shi, Laxman Dhulipala, David Eisenstat, Jakub {\L}{\k{a}}cki, and Vahab
  Mirrokni.
\newblock Scalable community detection via parallel correlation clustering.
\newblock {\em arXiv preprint arXiv:2108.01731}, 2021.

\bibitem[SM00]{ncut1}
Jianbo Shi and Jitendra Malik.
\newblock Normalized cuts and image segmentation.
\newblock {\em IEEE Transactions on pattern analysis and machine intelligence},
  22(8):888--905, 2000.

\bibitem[TCR12]{ganc}
Seyed~Salim Tabatabaei, Mark Coates, and Michael Rabbat.
\newblock Ganc: Greedy agglomerative normalized cut for graph clustering.
\newblock {\em Pattern Recognition}, 45(2):831--843, 2012.

\bibitem[WHFL15]{nmod-definition3}
Yashen Wang, Heyan Huang, Chong Feng, and Zhirun Liu.
\newblock Community detection based on minimum-cut graph partitioning.
\newblock In {\em International Conference on Web-Age Information Management},
  pages 57--69. Springer, 2015.

\bibitem[YD10]{nmod-definition2}
Linbin Yu and Chris Ding.
\newblock Network community discovery: Solving modularity clustering via
  normalized cut.
\newblock In {\em Proceedings of the Eighth Workshop on Mining and Learning
  with Graphs}, pages 34--36, 2010.

\bibitem[YS03]{ncut2}
Stella~X. Yu and Jianbo Shi.
\newblock Multiclass spectral clustering.
\newblock In {\em 9th {IEEE} International Conference on Computer Vision},
  pages 313--319. {IEEE} Computer Society, 2003.

\end{thebibliography}

\appendix

\section{Omitted Proofs}

\thmub*

\begin{proof}
  We start by proving the second inequality.
  Let $\Opt_G = \{C_1, \dots, C_k\}$ be a partition that maximizes $Q^0(\mathcal C)$. 
  It is easy to see that each cluster induces a connected subgraph of $G$.
  Indeed, if a subgraph $S$ induced by a cluster $C$ is disconnected, splitting $C$ into the connected components of $S$ increases the objective.

  We will argue that $\cquality(C_i)$ is bounded by two times the cost of the maximum spanning tree of $C_i$. The total cost of the spanning trees of all $C_i$ is at most $M$, which implies the result. Let $C=C_i$ be an arbitrary cluster of $\mathcal C$. We direct the edges with both end points in $C$ towards its end point with higher degree (ties are broken
  arbitrarily). Let $E_C$ be the set of edges obtained this way and let $E_C(v)$, $v\in C$, be the edges directed to $v$. The sets $E_C(v)$ form a partition of $E_C$.
  If for some $v\in C$ we have $2 \cdot \sum_{e\in E_C(v)} W(e) \ge \cquality(C)$ we argue that taking the tree formed by the set $E_C(v)$ has cost at least $\cquality(C)$.
  Indeed, the cost of this tree is $\sum_{e\in E_C(v)} W(e)$ which is at least $\cquality(C)/2$ by our assumption. 
  
  Thus, let us assume for sake of contradiction that for all $v\in C$ we have
  $2\cdot \sum_{e\in E_C(v)} W(e) < \cquality(C)$. Now observe that for $e\in E_C(v)$ we have $W(e) = \frac{1}{\deg(v)}$. Using this we obtain
  $$
  2\cdot \sum_{v\in C} |E_C(v)| \leq 2\cdot \sum_{v\in C} \frac{|E_C(v)|}{\deg(v)} \left( \sum_{v\in C} \deg(v)\right)  < \cquality(C) \cdot \sum_{v\in C} \deg(v). 
  $$
  This implies that
  $$
  \frac{2 \cdot \sum_{v\in C} |E_C(v)|} {\sum_{v\in C} \deg(v)} < \cquality(C),
  $$
  which is a contradiction since
  $$
  \frac{2 \cdot \sum_{v\in C} |E_C(v)|} {\sum_{v\in C} \deg(v)} = \frac{2|E(C)|}{\vol(C)} = \cquality(C).
  $$
  To prove the first inequality, we construct a clustering. We start by removing
  all edges of weight at most $1/\sqrt{n}$ from our maximum spanning forest. Let $F$ 
  be the resulting forest and let $M'$ be its weight. We observe that $M' \ge M-\sqrt{n}$.
  Now consider a greedy algorithm on $F$ that takes the most expensive edge $(u,v)$ and 
  creates a cluster of its endpoints and then removes all adjacent edges.
  Such a cluster has quality $\frac{2}{\deg(u)+\deg(v)}$ (which is at least the edge weight) and there are at most $\deg(u)+\deg(v)-2 \le 2 \sqrt{n}$ removed edges. 
  Since $(u,v)$ was an edge of maximum weight, we know that the weight of the removed 
  edges is at most $2\sqrt{n}$ times the weight of $(u,v)$.
  Hence, we obtain a solution with quality at least 
  $\frac{M'}{3\sqrt{n}} \ge\frac{M-\sqrt{n}}{3\sqrt{n}}$, which implies
  $Q^1(\Opt_G) \ge \frac{M-\sqrt{n}}{3\sqrt{n}}$. Rearranging the
  inequality implies the first inequality of the theorem.

\end{proof}

\lembipartite*

\begin{proof}
The algorithm runs in two steps. First, we color each vertex of $G$ red or blue independently 
at random. Then, we remove all edges whose both endpoints have the same color.
Denote the resulting graph by $G' = ((S\cup T), E', w)$, where we randomly select $S$ as being the set of red or the set of blue vertices and $T$ to be the other set. Now we remove all edges $(s,t)$ with
$s\in S$ and $t\in T$ that have $w(s) > w(t)$ and call the resulting set $E_H$.
For an edge $e=(u,v) \in E'$ let $X_e$ be the indicator random variable for the event
that $e \in E_H$. We claim that $\Pr[X_e=1]\ge 1/4$. Indeed, we observe that with probability $1/2$ the vertices $u$ and $v$ are colored differently and with probability at least $1/2$ the sides are chosen such
that $w(s) \le w(t)$. Now consider an optimal clustering $\Opt_G$ in $G$. We can write 
$q_G(\Opt_G) = \sum_{(u,v)\in E} \mathbbm{1}_{\Opt_G(u) = \Opt_G(v)}\frac{2}{w(\Opt_G(u))}$.
If $\mathcal{C}$ is a clustering of $V$ and $v \in V$ is a vertex, we use $\mathcal{C}(v)$ notation to denote the cluster containing $v$ in $\mathcal{C}$.
Hence, 
\begin{eqnarray*}
\E[q_H(\Opt_H)] & \ge & \E[q_H(\Opt_G)] \\
& = & \E[\sum_{(u,v)\in E_H} \mathbbm{1}_{\Opt_G(u)=\Opt_G(v)}\frac{2}{w(\Opt_G(u))}]\\
& = &  \sum_{e=(u,v)\in E} \E[X_e] \mathbbm{1}_{\Opt_G(u)=\Opt_G(v)}\frac{2}{w(\Opt_G(u))}\\
&\ge& \frac{1}{4}  q_G(\Opt_G)
\end{eqnarray*}
%
%
%
%
\end{proof}

\lemapxr*

\begin{proof}
The algorithm first applies the transformation of Lemma \ref{lem:bipartite} which
returns a bipartite graph $H$.
By the guarantees of Lemma \ref{lem:bipartite} we know that the expected quality of an optimal solution is at least $Q^0_w(\Opt_G)/4$, where $\Opt_G$ is an optimal clustering in $G$. 

Now observe that both the CVWAP problem and Lemma~\ref{lem:stars} consider partial clusterings $\mathcal{C} = \{C_1, \ldots, C_k\}$ such that $|C_i \cap S| = 1$ and $w(C_i \cap T) \leq 2w(C_i \cap S)$.
Let us call such partial clusterings \emph{restricted}.
Consider the restricted partial clustering $\mathcal C$ that is computed on input $H$ by our $\alpha$-approximation algorithm
for CVWAP as well as a restricted partial clustering
 $\Opt =\{O_1,\dots,O_k\}$ that maximizes 
$v(\Opt) = \sum_{i=1}^k \frac{2|O_i \cap T|}{3w(s_i)}$.
We clearly have $ v(\mathcal C) \ge \alpha \cdot v(\Opt)$ by the guarantee
of our approximation algorithm.
Observe that for any restricted $\mathcal{C}$ we have
$ 
	v(\mathcal{C}) = \sum_{i=1}^k \frac{2|C_i \cap T|}{3w(s_i)} \leq \sum_{i=1}^k \frac{2|C_i \cap T|}{w(s_i) + w(C_i \cap T)} = Q^0_w(\mathcal{C})
$
and
$
	v(\mathcal{C}) = \sum_{i=1}^k \frac{2|C_i \cap T|}{3w(s_i)} \geq \sum_{i=1}^k \frac{2|C_i \cap T|}{3(w(s_i) + w(C_i \cap T))} = Q^0_w(\mathcal{C})/3.
$
Furthermore, by Lemma~\ref{lem:stars}, we have $Q^0_w(\Opt) \geq Q^0_w(\Opt_H)/7$.
Hence,
$Q^0_w(\mathcal C) \ge v(\mathcal C) \ge \alpha \cdot v(\Opt) \ge \alpha \cdot Q^0_w(\Opt)/3 \ge \alpha \cdot Q^0_w(\Opt_H)/21$.
Thus, using linearity of expectation, Lemma \ref{lem:bipartite} it follows that the overall expected approximation ratio is $84 \cdot \alpha$.
\end{proof}

\end{document}